\documentclass[a4paper, 12pt,reqno]{amsart}
\usepackage[12pt]{extsizes}
\usepackage{amsmath,amssymb}
\usepackage{amscd}
\usepackage{geometry}
\usepackage{amsthm}
\usepackage{euscript}
\usepackage{latexsym}
\usepackage[matrix,arrow,curve,cmtip]{xy}
\usepackage[cp1251]{inputenc}
\usepackage[russian,english]{babel}
\usepackage{wrapfig}
\usepackage{mathrsfs}
\usepackage[pdftex]{graphicx}
\usepackage{keyval}
\usepackage{tikz}
\usetikzlibrary{arrows,shapes,snakes,automata,backgrounds,petri}
\usepackage{verbatim}
\usepackage[noend]{algorithmic}
\usepackage[symbol*]{footmisc}

\newtheorem{theorem}{Theorem}[section]
\newtheorem{lemma}{Lemma}[section]
\newtheorem{proposition}{Proposition}[section]
\newtheorem{corollary}{Corollary}[section]

\theoremstyle{definition}
\newtheorem{definition}[theorem]{Definition}
\newtheorem{example}[theorem]{Example}
\newtheorem{remark}[theorem]{Remark}
\newtheorem{construction}[theorem]{Construction}

\begin{document}

\pagestyle{plain}
\title{\bf PETRI NETS AND ITS POLYNOMIALS}

\maketitle
\begin{center}
Andrey Grinblat\footnote{\texttt{expandrey@mail.ru}} and Viktor Lopatkin\footnote{\texttt{wickktor@gmail.com}, please use this email for contacting.}
\end{center}

\begin{abstract}
  For every finite Petri net, we construct a commutative polynomial in two variables and with coefficients from the semiring of natural numbers. We also present an inverse construction and show that multiplication of polynomials correspondence to product of the corresponding Petri nets in the category of Petri nets with Winskel's morphisms. We endow the set of all Petri nets with Zariski topology.
\medskip


\end{abstract}

\section*{Introduction}
Petri nets are a tool for graphical and mathematical simulation, applicable to many systems. The are systems for describing and studying information processing systems that are characterized as being concurrent, asynchronous, distributed, parallel, nondeterministic, and/or stochastic. As a graphical tool, Petri nets can be used as a visual communication aid similar to flow charts, block diagrams, and networks. In addition, tokens are used in these nets to simulate the dynamics and concurrent activities of systems. As far as its being a mathematical tool, it is possible to set up state equations, algebraic equations, and other mathematical models governing the behavior of systems.

G. Winskel in \cite{W87} noticed that Petri nets can be viewed as certain 2-sorted algebras; it allows defining the concept of morphisms for Petri nets as homomorphisms of the corresponding algebras. Thus, the category of Petri nets is defined. The product of two Petri nets is defined in \cite{Win}.

It this paper, we consider every finite Petri net $N$ with injection map $\varphi:B \to \mathbb{N}$ on the set of its conditions, i.e., we label every condition by some natural number. For every pair $(N, \varphi)$ we construct a polynomial $P(N,\varphi) \in \mathbb{N}[x,y]$ in two variables and with coefficients from the semiring of natural numbers. Next, we present an inverse procedure, namely for every polynomial $P(x,y) \in \mathbb{N}[x,y]$, $P(0,0) \ne 0$ we construct the Petri net $\mathscr{N}(P)$. We show that multiplication of polynomials $P,Q\in \mathbb{N}[x,y]$, $P(0,0), Q(0,0) \ne 0$ correspondences to product of the Petri nets $\mathscr{N}(P) \times \mathscr{N}(Q)$ in the category of Petri nets with Winskel's morphisms. We also consider a Petri net $\mathscr{N}(P+Q)$ and show that it can be obtained from $\mathscr{N}(P)$ and $\mathscr{N}(Q)$ by ``attaching''. All of this enables us to introduce the Zariski topology on the set of Petri nets and we thus get a correspondence between prime ideals of $\mathbb{N}[x,y]$ and undecomposable Petri nets.

\section{The Category of Petri Nets}
In this section we recall some basic definitions of Petri nets theory and algebraic geometry. We essentially follow \cite {W87, Win} to define morphisms and product for Petri nets and we thus define Petri net category $\mathsf{PN}.$

\begin{definition}
A Petri net $N$ is a quadruple $(B,E, \mathrm{pre},\mathrm{post})$, where
\begin{itemize}
     \item[(1)] $B$, $E$ are disjoint finite sets of {\it conditions} and {\it events}, respectively,
     \item[(2)] $\mathrm{pre}:E \to 2^B$ is the {\it precondition} map such that $\mathrm{pre}(e)$ is nonempty for all $e \in E$,
     \item[(3)] $\mathrm{post}:E \to 2^B$ is the {\it postcondition} map such that $\mathrm{post}(e)$ is nonempty for all $e \in E$.
\end{itemize}
\end{definition}

Petri nets have a well-known graphical representation in which events are represented as boxes and conditions as circles with directed arcs between them (see fig.\ref{F1}).

\begin{figure}[h!]
\begin{tikzpicture}[node distance=1.3cm,>=stealth',auto]
  \tikzstyle{place}=[circle,thick,draw=blue!75,fill=blue!20,minimum size=6mm]
   \tikzstyle{transition}=[rectangle,thick,draw=black!75,
  			  fill=black!20,minimum size=4mm]
      \node [place, label=above:$b_0$]
                      (w1')                                                {};
    \node [place]     (c1') [below of=w1', label=below:$b_1$]              {};
    \node [place] (s1') [below of=c1',xshift=-7mm, label = left:$b_2$]      {};
    \node [place]
                      (s2') [below of=c1',xshift=7mm, label=right:$b_3$] {};
    \node [place]     (c2') [below of=s1',xshift=7mm, label=above:$b_4$]          {};
    \node [place]
                      (w2') [below of=c2', label=below:$b_5$]          {};
    \node [transition] (e1') [left of=c1'] {$e_1$}
      edge [pre,bend left]                  (w1')
      edge [post]                           (s1')
      edge [pre]                            (s2')
      edge [post]                           (c1');
    \node [transition] (e2') [left of=c2'] {$e_3$}
      edge [pre,bend right]                 (w2')
      edge [post]                           (s1')
      edge [pre]                            (s2')
      edge [post]                           (c2');
    \node [transition] (l1') [right of=c1'] {$e_2$}
      edge [pre]                            (c1')
      edge [pre]                            (s1')
      edge [post]                           (s2')
      edge [post,bend right] node[swap] {} (w1');
    \node [transition] (l2') [right of=c2'] {$e_4$}
      edge [pre]                            (c2')
      edge [pre]                            (s1')
      edge [post]                           (s2')
      edge [post,bend left]  node {}       (w2');
\end{tikzpicture}
\caption{An example of Petri net}\label{F1}
\end{figure}
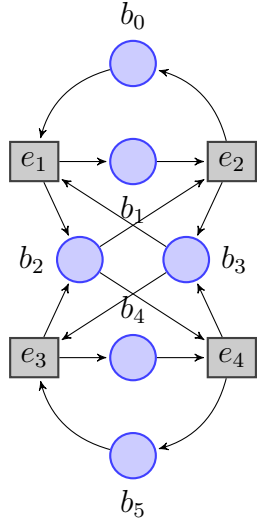

\begin{example}
  Let us consider the following Petri net $N$ which is shown in fig.\ref{F1}. We have $B = \{b_0,b_1,b_2,b_3,b_4,b_5\},$ $E = \{e_1,e_2,e_3,e_4\},$ $\mathrm{pre}(e_1) = \{b_0,b_3\}$, $\mathrm{post}(e_1) = \{b_1,b_2\}$, etc.
\end{example}

Let $N =(B,E, \mathrm{pre},\mathrm{post})$ be a Petri net with the events $E$. Define $E_*: = E \cup \{*\}$. We extend the pre and post condition maps to $*$ by taking
\[
\mathrm{pre}(*) = \varnothing, \qquad \mathrm{post}(*) = \varnothing.
\]

We will use the notation: whenever it does not cause confusion, we write $\mathstrut^\bullet e$ for the preconditions, $\mathrm{pre}(e)$ and $e^\bullet$ for the postcondition, $\mathrm{post}(e)$, of $e \in E_*$. We write $\mathstrut^\bullet e^\bullet$ for $\mathstrut^\bullet e  \cup e^\bullet$.

Now we aim to define the category of Petri nets \cite{W87, Win}.

\begin{definition}
Let $N =(B,E, \mathrm{pre},\mathrm{post})$ and $N' =(B',E', \mathrm{pre}',\mathrm{post}')$ be Petri nets. A morphism $(\beta,\eta):N \to N'$ consists of a relation $\beta \subseteq B \times B'$, such that $\beta^{\mathrm{op}}$ is a partial function $B' \to B$, and a partial function $\eta:E \to E'$ such that $\beta (\mathstrut^\bullet e) = \mathstrut^\bullet \eta(e)$, and $\beta(e^\bullet) = \eta(e)^\bullet.$ Thus, the diagrams
\[
 \xymatrix{
  E_* \ar@{->}[r]^{\mathrm{pre}} \ar@{->}[d]_{\eta} & 2^B \ar@{->}[d]^{2^\beta}\\
  E_*' \ar@{->}[r]_{\mathrm{pre}'}  & 2^{B'}
 }
 \qquad
 \xymatrix{
  E_* \ar@{->}[r]^{\mathrm{post}} \ar@{->}[d]_{\eta} & 2^B \ar@{->}[d]^{2^\beta}\\
  E_*' \ar@{->}[r]_{\mathrm{post}'}  & 2^{B'}
 }
\]
are commutative.
\end{definition}

\begin{proposition}\cite[Proposition 44]{Win}
Nets and their morphisms form a category $\mathsf{PN}$, in which the composition of two morphisms $(\beta,\eta):N\to N'$ and $(\beta',\eta'):N' \to N''$ is $(\beta\circ \beta', \eta \circ \eta'):N \to N''$ (composition in the left component being that of relations and in the right that of partial functions).
\end{proposition}

From the Proposition it follows that we have the following commutative diagrams
\[
 \xymatrix{
  E_* \ar@{->}[r]^\eta \ar@/^2pc/[rr]^{\eta \circ \eta'} \ar@{->}[d]_{\mathrm{pre}}& E_*'  \ar@{->}[r]^{\eta '} \ar@{->}[d]_{\mathrm{pre}'}& E_* \ar@{->}[d]_{\mathrm{pre}} \\
  2^B \ar@{->}[r]_{2^\beta} \ar@/_2pc/[rr]_{2^{\beta \circ \beta'}} & 2^{B'} \ar@{->}[r]_{2^{\beta'}} & 2^B
   }
\qquad
\xymatrix{
  E_* \ar@{->}[r]^\eta \ar@/^2pc/[rr]^{\eta \circ \eta'} \ar@{->}[d]^{\mathrm{post}}& E_*'  \ar@{->}[r]^{\eta '} \ar@{->}[d]^{\mathrm{post}'}& E_* \ar@{->}[d]^{\mathrm{post}} \\
  2^B \ar@{->}[r]_{2^\beta} \ar@/_2pc/[rr]_{2^{\beta \circ \beta'}} & 2^{B'} \ar@{->}[r]_{2^{\beta'}} & 2^B
   }
\]

Recall that a morphism $f:X \to Y$ in a category $\mathscr{C}$ is an isomorphism if it admits a two-sided inverse, meaning that there is another morphism $f^{-1}:Y\to X$ such that $f^{-1}\circ f = \mathbf{id}_X$ and $f\circ f^{-1} = \mathbf{id}_{Y}$, where $\mathbf{id}_X$ and $\mathbf{id}_Y$ are the identity morphisms of $X$ and $Y$, respectively. We use the standard notation $f:X \cong Y$.

\begin{remark}[{\bf isomorphism in the category $\mathsf{PN}$}]\label{Isom}
We thus get, that a morphism $(\beta,\eta):N \to N'$ is an isomorphism, in the category $\mathsf{PN}$, if there is another morphism $(\beta^{-1},\eta^{-1})$ such that the following diagrams
\[
 \xymatrix{
  & E_* \ar@{->}[rr]^{\mathrm{pre}}  \ar@{->}[ld]_{\eta}  \ar@{->}[dd]^(.3){\mathbf{id}_{E_*}} && 2^B  \ar@{->}[ld]_{2^\beta}  \ar@{->}[dd]_{2^{\mathbf{id}_B}}\\
  E'_*  \ar@{->}[rr]_(0.7){\mathrm{pre}'}  \ar@{->}[rd]_{\eta^{-1}} && 2^{B'}  \ar@{->}[rd]_(0.4){2^{\beta^{-1}}}& \\
  & E_*  \ar@{->}[rr]_{\mathrm{pre}} && 2^B
 }
 \qquad
  \xymatrix{
  & E_* \ar@{->}[rr]^{\mathrm{post}}  \ar@{->}[ld]_{\eta}  \ar@{->}[dd]^(.3){\mathbf{id}_{E_*}} && 2^B  \ar@{->}[ld]_{2^{\beta^{-1}}}  \ar@{->}[dd]_{2^{\mathbf{id}_B}}\\
  E'_*  \ar@{->}[rr]_(0.7){\mathrm{post}'}  \ar@{->}[rd]_{\eta^{-1}} && 2^{B'}  \ar@{->}[rd]_(0.4){2^{\beta^{-1}}}& \\
  & E_*  \ar@{->}[rr]_{\mathrm{post}} && 2^B
 }
\]
are commutative. It follows that, in the case when $N$ and $N'$ are finite then an isomorphism between them is a pair $(\beta,\eta)$ of two bijections such that the aforementioned diagrams are commutative.
\end{remark}

In \cite{Win}, is was defined the product of Petri nets. The product of nets and its behavior are more straightforward, and, as is expected, correspond to a synchronization operation on nets.

\begin{definition}[{\bf Product of Petri nets}]\label{parallproduct} Let $N_1 =(B_1,E_{1*}, \mathrm{pre}_1,\mathrm{post}_1)$ $N_2 =(B_2,E_2, \mathrm{pre}_2,\mathrm{post}_2)$ be Petri nets. {\it Their product} $N = N_1 \times N_2:= (B,E_*,\mathrm{pre},\mathrm{post});$ it has the events $E : = E_{1*} \times_* E_{2*}$, the product in $\mathsf{Set}_*$ with the projections $\pi_1:E_* \to_* E_{1*}$ and $\pi_2:E_* \to_* E_{2*}$. Its conditions have the form $B: = B_1 \sqcup B_2$, the disjoint union of $B_1$ and $B_2$. Define $\rho_1$ to be the opposite relation to the injection $(\rho_1)^{\mathrm{op}}:B_1 \to B$. Define $\rho_2$ similarly. Define the pre and post conditions of an event $e$ in the product in terms of its pre and post conditions in the components by
\begin{align*}
& \mathrm{pre}(e): = (\rho_1)^{\mathrm{op}}[\mathrm{pre}_1(\pi_1(e))] + (\rho_2)^{\mathrm{op}}[\mathrm{pre}_2(\pi_2(e))]\\
& \mathrm{post}(e): = (\rho_1)^{\mathrm{op}}[\mathrm{post}_1(\pi_1(e))] +  (\rho_2)^{\mathrm{op}}[\mathrm{post}_2(\pi_2(e))].
\end{align*}
\end{definition}

\begin{figure}[h!]
\begin{center}
\begin{tikzpicture}[node distance=1.3cm,>=stealth',bend angle=25,auto]
  \tikzstyle{place}=[circle,thick,draw=blue!75,fill=blue!20,minimum size=6mm]
  \tikzstyle{red place}=[place,draw=red!75,fill=red!20]
  \tikzstyle{transition}=[rectangle,thick,draw=black!75,
  			  fill=black!20,minimum size=6mm]
  \begin{scope}
   \node [place, label=right:$0$] (v1)                                 {};
   \node [transition] (ee1) [above of=v1] {$a$}
     edge [pre, bend left]   (v1)
     edge [post, bend right] (v1);
    \end{scope}
\begin{scope}[xshift=2cm]
 \node [place, label = right:$1$] (v1')                 {};
 \node [transition] (b) [above of = v1'] {$b$}
   edge [post] (v1');
 \node [transition] (c) [below of = v1'] {$c$}
   edge [pre] (v1');
\end{scope}
\begin{scope}[xshift=6cm]
 \node [place, label = left:$0$] (w0)  {};
 \node[place] (w1) [right of=w0,xshift=3cm,label=right:$1$] {};
 \node [transition] (a*) [above of = w0] {$(a,*)$}
   edge [pre, bend left] (w0)
   edge [post, bend right] (w0);
 \node[transition] (*b) [above of= w1] {$(*,b)$}
   edge [post] (w1);
 \node [transition] (ab) [left of=*b,xshift=-1cm] {$(a,b)$}
   edge [pre,bend angle = 15, bend left] (w0)
   edge [post,bend angle = 15, bend right] (w0)
   edge [post,bend angle = 15, bend left] (w1);
 \node [transition] (*c) [below of=w1] {$(*,c)$}
   edge [pre] (w1);
 \node [transition] (ac) [left of=*c,xshift=-1cm] {$(a,c)$}
   edge [pre,bend angle = 15, bend right] (w1)
   edge [pre,bend angle = 15, bend left] (w0)
   edge [post,bend angle = 15, bend right] (w0);
 \end{scope}
\end{tikzpicture}
\end{center}
\caption{The product of two Petri nets is shown.}
\end{figure}
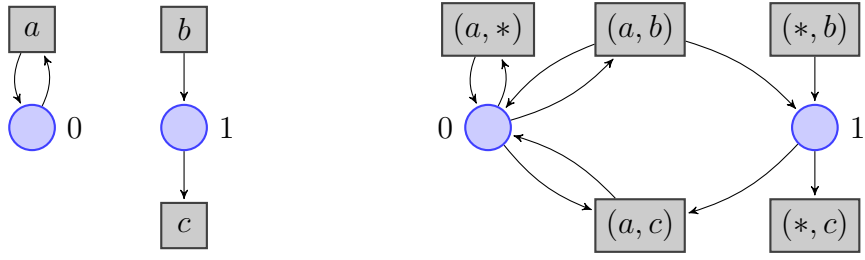

In the case $N = N_1 \times N_2$ we say $N_1,N_2$ are factors of $N$, $N$ is {\it decomposable} if $N_1$, $N_2$ are not empty Petri nets and {\it undecomposable} otherwise. Further, a factor $N_1$ of some Petri net $N$ is called {\it prime factor} if $N_1$ is not empty and it is undecomposable.

\section{Net's Polynomials}
In this section we construct for every Petri net $N$ a polynomial $P\in \mathbb{N}[x,y]$ and we shall shown that two isomorphic Petri nets, in the category $\mathsf{PN}$, can have the same polynomial.

We start from the following procedure allows to construct a polynomial $P(N) \in \mathbb{N}[x,y]$  from a given Petri net $N$.

\begin{construction}\label{construction1}
Let $N = (B,E_*,\mathrm{pre}, \mathrm{post})$ be a Petri net and $\varphi:B \to \mathbb{N}$ an injection. Set $P(N,\varphi):=\sum\limits_{e\in E_*}x^{i(e)}y^{j(e)}\in \mathbb{N}[x,y],$ where $i(e): = \sum\limits_{b \in {\mathstrut^\bullet e}}2^{\varphi(b)}$, $j(e): = \sum\limits_{b \in e^\bullet}2^{\varphi(b)}$ and we put $e \mapsto x^0y^0 = 1$ iff $\mathstrut^\bullet e^\bullet = \varnothing$, in particular $* \mapsto 1$.
\end{construction}

\begin{example}
Let us consider the following Petri net $N$ (see fig.\ref{ex1}), with $B = \{b_0,b_1,b_2,b_3,b_4\}$, $E_* = \{e_1, e_2,e_3,e_4,e_5,e_6 ,*\}$.

Set $\varphi(b_0)=0$, $\varphi(b_1) =1$, $\varphi(b_2) =2$, $\varphi(b_3) =3$ and $\varphi(b_4) =4$. We then get: $e_1 \mapsto x^{2^0}y^{2^2} = xy^4$, $e_2 \mapsto x^{2^0}y^{2^2} = xy^4$, $e_3 \mapsto x^{2^1}y^{2^3} = x^2y^{8}$, $e_4 \mapsto x^{2^1}y^{2^3} = x^2y^{8}$, $e_5 \mapsto x^{2^1}y^{2^3} = x^2y^{8}$, $e_6 \mapsto x^{2^2+2^3}y^{2^4} = x^{12}y^{16}$ and $* \mapsto 1$. We thus obtain $P(N,\varphi) = 2xy^4 + 3 x^2y^8 + x^{12}y^{16} + 1$.
\end{example}

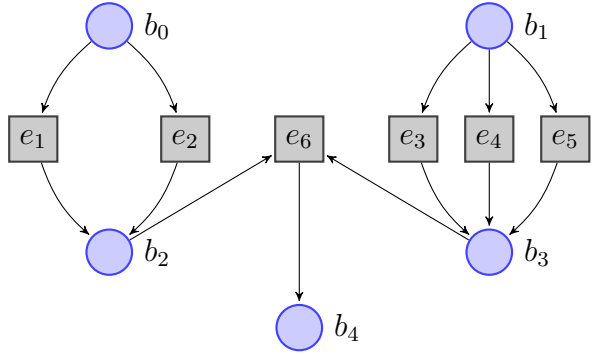
\begin{figure}[h!]
\begin{center}
\begin{tikzpicture}[node distance=1.3cm,>=stealth',bend angle=15,auto]
  \tikzstyle{place}=[circle,thick,draw=blue!75,fill=blue!20,minimum size=6mm]
  \tikzstyle{red place}=[place,draw=red!75,fill=red!20]
  \tikzstyle{transition}=[rectangle,thick,draw=black!75,
  			  fill=black!20,minimum size=6mm]
   \node[place,label=right:$b_0$] (0) at (0,4) {};
   \node[place,label=right:$b_1$] (1) at (5,4) {};
   \node[place,label=right:$b_2$] (2) at (0,1) {};
   \node[place,label=right:$b_3$] (3) at (5,1) {};
   \node[place,label=right:$b_4$] (4) at (2.5,0) {};
   \node[transition] (a) at (-1,2.5) {$e_1$} edge [post, bend right] (2) edge [pre, bend left] (0);
   \node[transition] (b) at (1,2.5) {$e_2$} edge [post, bend left] (2) edge [pre, bend right] (0);
   \node[transition] (c) at (4,2.5) {$e_3$} edge [post, bend right] (3) edge [pre, bend left] (1);
   \node[transition] (d) at (5,2.5) {$e_4$} edge [post] (3) edge [pre] (1);
   \node[transition] (e) at (6,2.5) {$e_5$} edge [post, bend left] (3) edge [pre, bend right] (1);
   \node[transition] (f) at (2.5,2.5) {$e_6$} edge [pre] (2) edge [pre] (3) edge [post] (4);
\end{tikzpicture}
\end{center}
\caption{Here the Petri net and its polynomial $P(N,\varphi) = 2xy^4 + 3 x^2y^8 +  x^{12}y^{16} + 1$ are shown.}\label{ex1}\end{figure}

\begin{lemma}
  Let $N = (B,E,\mathrm{pre},\mathrm{post})$, $N = (B',E',\mathrm{pre}',\mathrm{post}')$ be Petri nets; if there is an isomorphism $(\beta,\eta):N \xrightarrow{\cong} N'$ (in the category $\mathsf{PN}$, see Remark \ref{Isom}), then $P(N,\varphi) = P(N',\varphi\circ \beta^{-1})$, where $\varphi:B \to \mathbb{N}$ is an arbitrary injection.
\end{lemma}
\begin{proof}
Let $P(N,\varphi) = \sum\limits_{e \in E_*}x^{i(e)}y^{j(e)}$ and $\beta(b) = b'$, $\eta(e) = e'$. Since the pair $(\beta,\eta)$ is the isomorphism $N \cong N'$ (see Remark \ref{Isom}) then it implies
\begin{align*}
  & i(e) :=  \sum\limits_{b \in \mathstrut^\bullet e} 2^{\varphi(b)} = \sum\limits_{\beta^{-1}(b') \in \mathstrut^\bullet \eta^{-1}(e)} 2^{\varphi (\beta^{-1} (b'))} = \sum\limits_{b' \in \mathstrut^\bullet e'} 2^{\varphi (\beta^{-1}(b'))}:=i(e')\\
  & j(e) :=  \sum\limits_{b \in e^\bullet} 2^{\varphi(b)} = \sum\limits_{\beta^{-1}(b') \in \eta^{-1}(e)^\bullet} 2^{\varphi (\beta^{-1} (b'))} = \sum\limits_{b' \in e'^\bullet} 2^{\varphi (\beta^{-1}(b'))}:=j(e').
\end{align*}

We thus get the following polynomial $P(N',\varphi\circ \beta^{-1}) = \sum\limits_{e' \in E'_*}x^{i(e')}y^{j(e')}$ which is obviously equal to $P(N,\varphi)$. This completes the proof.
\end{proof}

We explicitly describe an inverse procedure to Construction \ref{construction1}. This procedure allows us to construct a Petri net $N(P)$, with a concrete injection $\iota: \{\mbox{the set of condtions of $N$}\} \to \mathbb{N}$, from a given polynomial $P\in \mathbb{N}[x,y]$ with $P(0,0) \ne 0$. We shall further show that for every Petri net $N'$ (with an injection $\varphi$), the Petri net $N(P(N',\varphi))$ can be identified with $N$ (see Proposition \ref{Prop=}).

At first, we start with some notations. Let $k= \varepsilon_0(k)2^0 +\varepsilon_1(k)2^1 + \cdots + \varepsilon_{\ell_k}(k)2^{\ell_k}$ be a binary decomposition of $k \in \mathbb{N}$. Set $\tau(k): = \{t:\varepsilon_t(k) =1\}$ if $k>0$ and $\tau(0) = \varnothing$ in otherwise.

\begin{construction}\label{con2}
Let $I\subset \mathbb{N} \times \mathbb{N}$ be a finite set such that $(0,0) \in I$ and let $P = P(x,y) = \sum\limits_{(i,j) \in I} a_{i,j}x^iy^j \in \mathbb{N}[x,y]$ be a polynomial such that $a_{0,0} \ne  0$. For every pair $(i,j) \in I$ we form the following set $E_{i,j}:=\left\{e_1^{(i,j)},\ldots,e_{a_{i,j}}^{(i,j)} \right\}$ of some elements. Set $\mathscr{N}(P): = \Bigl(B,E_*, \mathrm{pre}, \mathrm{post}\Bigr)$, where
\[
 B =\tau(P):= \bigcup\limits_{(i,j) \in I} \{\tau(i), \tau(j)\}, \quad * := e_{a_{0,0}}^{(0,0)}, \quad E_* = \bigcup\limits_{(i,j)\in I}E_{i,j},
 \]
the maps $\mathrm{pre}, \mathrm{post}:E \to 2^B$ and the injection $\iota:B \to \mathbb{N}$ are defined as follows
\[
 \mathrm{pre}\left( e_{k}^{(i,j)}\right) =\tau(i), \qquad  \mathrm{post}\left(e_{k}^{(i,j)}\right)= \tau(j), \quad \iota(\tau(i)) = \tau(i).
\]

\end{construction}

\begin{figure}[h!]
\begin{center}
\begin{tikzpicture}[node distance=1.3cm,>=stealth',bend angle=25,auto]
  \tikzstyle{place}=[circle,thick,draw=blue!75,fill=blue!20,minimum size=6mm]
  \tikzstyle{red place}=[place,draw=red!75,fill=red!20]
  \tikzstyle{transition}=[rectangle,thick,draw=black!75,
  			  fill=black!20,minimum size=6mm]
  \node[place,label=below:$0$] (0) at (0,0) {};
  \node[place,label=below:$1$] (1) at (4,0) {};
  \node[transition] (e101) at (-2,0) {$e_{1}^{(0,1)}$} edge [post] (0);
  \node[transition] (e133) at (2,0) {$e_{1}^{(3,3)}$} edge [post,bend left] (0) edge [pre, bend right] (0) edge[pre, bend right] (1) edge[post, bend left] (1);
  \node[transition] (e220) at (4,1.5) {$e_{2}^{(2,0)}$} edge [pre] (1);
  \node[transition] (e120) at (6,0) {$e_{1}^{(2,0)}$} edge [pre] (1);
  \node[transition] (*) at (8,1.5) {$*$};
  \node[transition] (e100) at (8,0) {$e_{1}^{(0,0)}$};
\end{tikzpicture}
\end{center}
\caption{For the given polynomial $P = x^3y^3 + 2x^2 + y + 2$, the corresponding Petri net is shown.}\label{ex2}\end{figure}
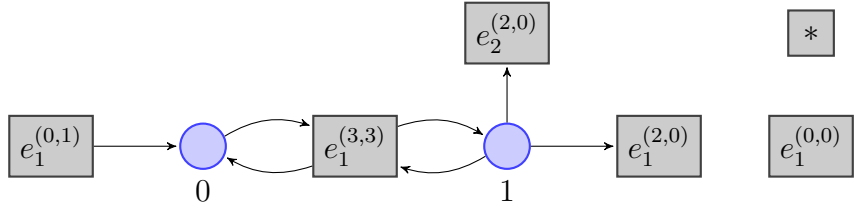

\begin{example}
For the following polynomial $P = x^3y^3 + 2x^2 + y + 2$, let us construct the corresponding Petri net $(B,E_*,\mathrm{pre},\mathrm{post})$. We obtain:
\begin{eqnarray*}
B &=& \tau(3) \cup \tau(2) \cup \tau(1) \cup \tau(0) \\
 &=&\{1,0\} \cup \{1\} \cup \{0\} \cup \varnothing = \{1,0\},
\\
E_* &=& E_{3,3} \cup E_{2,0} \cup E_{0,1} \cup E_{0,0} \\
 &=& \left\{e_1^{(3,3)}\right\} \cup \left\{e_1^{(2,0)}, e_2^{(2,0)}\right\} \cup \left\{e_1^{(0,1)}\right\} \cup \left\{e_1^{(0,0)},e_2^{(0,0)}=*\right\},
\end{eqnarray*}
and $\iota(0) = 0$, $\iota(1) = 1$. The maps $\mathrm{pre},\mathrm{post}:E \to 2^B$ are defined as it shown in fig.\ref{ex2}.
\end{example}

\begin{proposition}\label{Prop=}
Let $N = (B,E_*,\mathrm{pre},\mathrm{post})$ be a Petri net and $\varphi:B \to \mathbb{N}$ be an injection. There exists an isomorphism $N \cong \mathscr{N}(P(N,\varphi))$, in the category $\mathsf{PN}$.
\end{proposition}
\begin{proof}
Let $\mathscr{N}(P(N,\varphi)) = \bigl(\widetilde{B},\widetilde{E}_*, \widetilde{\mathrm{pre}}, \widetilde{\mathrm{post}}\bigr)$. From Construction \ref{con2} it follows that $\widetilde{B} = \tau(P(N,\varphi))=  \{n \in \mathbb{N}:n \in \mathrm{Im}(\varphi) \} = \mathrm{Im}(\varphi)$, i.e., $\widetilde{B} = \mathrm{Im}(\varphi)$, but since $\varphi$ is injection, we thus get the following bijection $\beta=\varphi^{-1}:\widetilde{B} \cong B:\varphi =\beta^{-1}.$ Further, from Construction \ref{construction1} it follows that $\sum\limits_{(i,j)\in I}a_{i,j} = |E_*|$ it implies $|\widetilde{E}_*| = \sum\limits_{(i,j) \in I}a_{i,j} = |E_*|,$ and thus there exist bijections between these sets. Let us choose a bijections (say) $\eta: \widetilde{E_*} \cong E_*$ satisfies the following conditions
\[
\begin{cases} \mathstrut^\bullet \eta(e_{k}^{(i,j)}) = \tau(i) \in \widetilde{B},\\
 \eta(e_{k}^{(i,j)})^\bullet = \tau(j) \in \widetilde{B},
 \end{cases}
 \quad
 \begin{cases} \mathstrut^\bullet \eta^{-1}\left(e\right) = \varphi^{-1}(\tau(i)) \in B,\\
 \eta\left(e\right)^\bullet = \varphi^{-1}(\tau(j)) \in B,
 \end{cases}
\]
for every $\widetilde{E} \ni e_{k}^{(i,j)} \leftrightarrow e \in E_*$, we then obtain the needed isomorphism $(\beta, \eta): N \cong \mathscr{N}(P(N,\varphi))$. This completes the proof.
\end{proof}

\section{Algebraic Operations on Petri Nets}
In this section we consider sum and multiplication of polynomials from the semiring $\mathbb{N}[x,y]$ in the context of Petri nets: we prove that multiplication of two polynomials $P$, $P'$, such that $P(0,0),P'(0,0) \ne 0$, corresponding to product of the Petri nets $\mathscr{N}(P)$, $\mathscr{N}(P')$; and the sum $P+P'$ correspondences to operation, on the set of Petri net, which was described in \cite{Kotov}. We shall see that this way implies a criteria for a decomposition Petri nets (in the sense of Definition \ref{parallproduct}).

\subsection{Multiplication of net's polynomials}
Here we take an interest in the multiplication of two polynomials $P_1,P_2 \in \mathbb{N}[x,y]$ in the context of corresponding Petri nets $\mathscr{N}(P_1)$, $\mathscr{N}(P_2)$. We shall show that it correspondences to the product of Petri nets in category $\mathsf{PN}$.

First of all we need the following technical but useful result.

\begin{lemma}\label{Lemma2}
Let $P = P_1 \cdot P_2 \in \mathbb{N}[x,y]$; then $\tau(P_1) \cap \tau(P_2) = \varnothing$ if and only if $\tau(P) = \tau(P_1) \cup \tau(P_2)$.
\end{lemma}
\begin{proof}
It obviously suffices to prove that for any two integers $i_1,i_2 \in \mathbb{N}$, $\tau(i_1+i_2) = \tau(i_1) \cup \tau(i_2)$ if and only if $\tau(i_1) \cap \tau(i_2) = \varnothing$.

Let $i_1 = \varepsilon_0^{(1)}2^0 + \varepsilon_1^{(1)}2^1 + \ldots +\varepsilon_0^{(1)}2^{\ell_{i_1}}$ and $i_2 = \varepsilon_0^{(2)}2^0 + \varepsilon_1^{(2)}2^1 + \ldots +\varepsilon_0^{(2)}2^{\ell_{i_2}}$. Assume that $\tau(i_1)\cap \tau(i_2) = \varnothing$, it follows that $i_1+i_2 = \varepsilon_0 2^0 + \varepsilon_12^1 + \ldots +\varepsilon_\ell2^{\ell}$, where $\varepsilon_j = \mathrm{max}\left\{\varepsilon_j^{(1)},\varepsilon_j^{(2)}\right\}$, $1 \le j \le \mathrm{max}\{\ell _{i_1}, \ell_{i_2}\} = \ell,$ i.e., $\tau(i_1 + i_2) = \tau(i_1) \cup \tau(i_2).$ Conversely, let $\tau(i_1 + i_2) = \tau(i_1) + \tau(i_2)$ and suppose that $\tau(i_1) \cap \tau(i_2) \ne \varnothing$. For every $n \in \tau(i_1) \cap \tau(i_2)$, we obviously have $\varepsilon_n^{(1)} + \varepsilon_n^{(2)} = 1+1 = 0\bmod (2)$, it follows that $n \notin \tau(i_1+i_2)$, but $n \in \tau(i_1) \cap \tau(i_2) \subseteq \tau(i_1) \cup \tau(i_2) =  \tau(i_1+i_2)$, i.e., we get a contradiction.
\end{proof}

\begin{theorem}\label{theorem}
Let $(N,\varphi)=\bigl((B,E_*, \mathrm{pre},\mathrm{post}),\varphi\bigr)$ be a Petri net with an injection $\varphi:B \to \mathbb{N}$, let $ P_1, P_2 \in \mathbb{N}[x,y]$ such that $\tau(P_1)\cap \tau(P_2) = \varnothing$. Then $P(N,\varphi) = P_1 \cdot P_2$ if and only if $\mathscr{N}(P_1) \times \mathscr{N}(P_2) \cong N$, in the category $\mathsf{NP}$.
\end{theorem}
\begin{proof} Let $P(N,\varphi) = \sum\limits_{(i,j) \in I}a_{i,j}x^iy^j$, $P_1 = \sum\limits_{(i,j) \in I_1}a_{i,j}^{(1)}x^iy^j$, $P_2 = \sum\limits_{(i,j) \in I_2}a_{i,j}^{(2)}x^iy^j$, here $I,I_1,I_2$ are finite subsets of $\mathbb{N}\times \mathbb{N}$, such that $(0,0) \in I,I_1,I_2$. Further, let us set $\mathscr{N}(P(N,\varphi)) = \bigl(\widetilde{B},\widetilde{E}_*, \mathrm{pre}, \mathrm{post}\bigr)$, $\mathscr{N}(P_1) = \bigl(B_1,E_{1*}, \mathrm{pre}, \mathrm{post}\bigr)$ and $\mathscr{N}(P_1) = \bigl(B_2,E_{2*}, \mathrm{pre}, \mathrm{post}\bigr)$.

(1) Assume that $P(N,\varphi)=P_1\cdot P_2$. By Proposition \ref{Prop=}, there exists an isomorphism $\mathscr{N}(P(N,\varphi)) \cong N$ and by Lemma \ref{Lemma2} we get $\tau(P(N,\varphi)) = \tau(P_1) \cup \tau(P_2)$. We have also seen (see the proof of Proposition \ref{Prop=}) that $\widetilde{B} = \tau(P(N,\varphi))$. Fix $(i,j) \in I$, it is clear that $a_{i,j} = \sum\limits_{\substack{i_1 + i_2 = i \\ j_1 + j_2 = j}}a^{(1)}_{i_1,j_1}a^{(2)}_{i_2,j_2}$,  here $(i_1,j_1)\in I_1$, $(i_2,j_2)\in I_2$. It follows from the preceding discussion that
\[
\left|\widetilde{E_*}\right| = |E_*| =  \sum_{(i,j) \in I} a_{i,j} = \sum\limits_{(i,j) \in I_1}a_{i,j}^{(1)}\sum\limits_{(i,j) \in I_2}a_{i,j}^{(2)} = \left|\widetilde{E_{1*}}\right| \times \left|\widetilde{E_{2*}}\right|.
\]

Hence, there exist bijections $\widetilde{E_{*}} \cong E_* \cong \widetilde{E_{1*}} \times \widetilde{E_{2*}}$. Next, let us consider the following map
\[
E_{1*}\times E_{2*} \ni \bigl(x^{i_1}y^{j_1},x^{i_2}y^{j_2} \bigr) \mapsto  x^{i_1+i_2}y^{j_1+j_2} \in E_*,
\]
since $\tau(P_1) \cap \tau(P_2) = \varnothing$, we see that this map is an injection. Further, from the fact $E_* \cong E_{1*} \times E_{2*}$ it follows that this map is a bijection. Hence we get
\[
b \in \mathstrut^\bullet(e_1,e_2) \Longleftrightarrow b \in \tau(i_1) \cup \tau(i_2) \Longleftrightarrow b \in \tau(i_1 + i_2) \Longleftrightarrow b \in \mathstrut^\bullet e,
\]
it implies that $\mathscr{N}(P_1) \times \mathscr{N}(P_2) \cong N$.

(2) Conversely, assume that $\mathscr{N}(P_1) \times \mathscr{N}(P_2) \cong N$. We have $P(N,\varphi): = \sum\limits_{e \in E_*}x^{i(e)}y^{j(e)}$, where $i(e): = \sum\limits_{b \in \mathstrut^\bullet e}2^{\varphi(b)}$ and $j(e): = \sum\limits_{b\in e^\bullet} 2^{\varphi(b)}$. Because $N \cong \mathscr{N}(P_1) \times \mathscr{N}(P_2)$, it follows from Definition \ref{parallproduct} that
\begin{eqnarray*}
 i(e) &=& i((e_1,e_2)) = \sum\limits_{b \in \mathstrut^\bullet e}2^{\varphi(b)} =  \sum\limits_{b\in \mathstrut^\bullet (e_1,e_2)}2^{\varphi(b)}\\
 &=& \sum\limits_{b_1 \in \mathstrut^\bullet e_1}2^{\varphi(b_1)} + \sum\limits_{b_2 \in \mathstrut^\bullet e_2}2^{\varphi(b_2)} = i(e_1) + i(e_2),
\end{eqnarray*}
here $e_1 \in E_{1*}$ and $e_2 \in E_{2*}$. In the same way one can easy get $j(e) = j((e_1,e_2)) = j(e_1)+ j(e_2)$. We thus obtain
\begin{eqnarray*}
  P(N,\varphi) &=& \sum\limits_{e \in E_*} x^{i(e)}y^{j(e)} = \sum\limits_{e_1\in E_{1*}}\sum\limits_{e_2 \in E_{2*}}x^{i(e_1) + i(e_2)} y^{j(e_1)+ j(e_2)}\\
  &=& \sum\limits_{e_1 \in E_{1*}}x^{i(e_1)}y^{j(e_1)} \sum\limits_{e_2 \in E_{2*}}x^{i(e_2)}y^{j(e_2)} = P_1\cdot P_2.
\end{eqnarray*}

Finally, note that Definition \ref{parallproduct} implies that $B$ has to be the disjoint union $B_1 \sqcup B_2$ of $B_1$ and $B_2$. It shows that $\tau(P_1) \cap \tau(P_2) = \varnothing.$ This completes the proof.
\end{proof}

\begin{corollary}\label{corol1}
  A Petri net $N = (B,E,\mathrm{pre},\mathrm{post})$ is decomposable, i.e., it can be represented as a product of two Petri nets in the category $\mathsf{PN}$, if and only if the polynomial $P(N,\varphi)$, for an arbitrary choice of an injection $\varphi:B \to \mathbb{N}$, is decomposable over $\mathbb{N}[x,y]$, $P(N,\varphi) = P_1 \cdot P_2$ and $\tau(P_1) \cap \tau(P_2) = \varnothing$.
\end{corollary}
\begin{proof}
  This immediately follows from Theorem \ref{theorem}.
\end{proof}

\begin{figure}[h!]
\begin{center}
\begin{tikzpicture}[node distance=1.3cm,>=stealth',bend angle=25,auto]
  \tikzstyle{place}=[circle,thick,draw=blue!75,fill=blue!20,minimum size=6mm]
  \tikzstyle{red place}=[place,draw=red!75,fill=red!20]
  \tikzstyle{transition}=[rectangle,thick,draw=black!75,
  			  fill=black!20,minimum size=6mm]
  \node[place,label=below:$b_0$] (0) at (0,0) {};
  \node[place,label=below:$b_1$] (1) at (4,0) {};
  \node[transition] (a) at (-2,0) {$a$} edge [pre] (0);
  \node[transition] (b) at (2,0) {$b$} edge [pre] (0) edge[post] (1);
  \node[transition] (c) at (6,0) {$c$} edge [post] (1);
  \end{tikzpicture}
\end{center}
\caption{The Petri net $N$.}\label{parall(1)}
\end{figure}
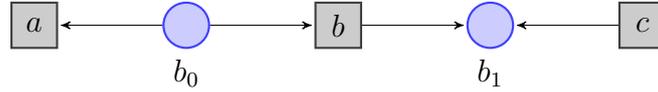

\begin{example}
Let us consider the following Petri net $N$ which is shown in fig.\ref{parall(1)}. Let us find out whether this Petri net can be decomposable.

Set $\varphi(b_0) = 0$ and $\varphi(b_1) = 1$. We thus get $a \mapsto x^{2^0} = x$, $b \mapsto x^{2^0}y^{2^1} = xy^2,$ $ \mapsto y^{2^1} = y^2$ and $ * \mapsto 1.$ It implies that $P(N,\varphi) = x + xy^2 + y^2+1$. It is not hard to see that $P(N,\varphi) = x + xy^2 + y^2+1 = (x+1)(y^2+1)$. Let $P_1 = x+1$ and $P_2 = y^2+1$. We see that $\tau(P_1) \cap \tau(P_2) = \varnothing.$ Using Construction \ref{con2} we construct the Petri nets $\mathscr{N}(P_1)$ and $\mathscr{N}(P_2)$ (see fig.\ref{parall(2)}). It is easy to see that $N = \mathscr{N}(P_1) \times \mathscr{N}(P_2)$.
\begin{figure}[h!]
\begin{center}
\begin{tikzpicture}[node distance=1.3cm,>=stealth',bend angle=25,auto]
  \tikzstyle{place}=[circle,thick,draw=blue!75,fill=blue!20,minimum size=6mm]
  \tikzstyle{red place}=[place,draw=red!75,fill=red!20]
  \tikzstyle{transition}=[rectangle,thick,draw=black!75,
  			  fill=black!20,minimum size=6mm]
  \node[place,label=right:$0$] (0) at (0,0) {};
  \node[place,label=right:$1$] (1) at (4,0) {};
  \node[transition] (e1) at (0,-1.5) {$e_1$} edge [pre] (0);
  \node[transition] (e2) at (4,-1.5) {$e_2$} edge [post] (1);
  \end{tikzpicture}
\end{center}
\caption{Two Petri nets $\mathscr{N}(P_1)$ (at left) and $\mathscr{N}(P_1)$ (at right) that correspond to the polynomials $P_1 = x+1$ and $P_2 = y^2+1$, respectively.}\label{parall(2)}\end{figure}
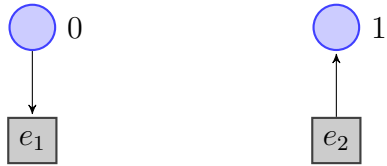
\end{example}

\subsection{Sum of net's polynomials}
Here we consider sum $P_1+P_2$ of two polynomials $P_1, P_2 \in \mathbb{N}[x,y]$ in the context of Petri nets, i.e., we want to know what a Petri net arises from a sum of two net's polynomials.

We start with the following example.

\begin{example}
Let $N_1 = \{B_1,E_{1*},\mathrm{pre}_1,\mathrm{post}_1\}$, $N_2 = \{B_2,E_{2*},\mathrm{pre}_2,\mathrm{post}_2\}$ be Petri nets which are shown in fig.\ref{last1}. We have $B_1 = \{b_{11}, b_{12}\}$, $E_{1*} = \{a,*_1\}$, $B_2 = \{b_{21}, b_{22}, b_{23}\}$ and $E_{2*} = \{b,*_2\}$.

Let injections $\varphi_1:B_1 \to \mathbb{N}$, $\varphi_2: B_2 \to \mathbb{N}$ be given by $\varphi_1(b_{11}) = 1$, $\varphi_1(b_{12}) = 2$ and $\varphi_2(b_{21}) = 2$, $\varphi_2(b_{22}) = 3$, $\varphi_2(b_{23}) = 4$. Then, using Construction \ref{construction1}, we obtain $P(N_1,\varphi_1) = xy^4+1$ and $P(N_2,\varphi_2) = x^4y^{24}+1.$

We now aim to construct a Petri net correspondences to the sum of these polynomials. Let $P:= x^4y^{24}+ xy^4  + 2$. Using Construction \ref{con2} and taking account that $24 = 11000_2$, $4 = 100_2$, we obtain
\begin{eqnarray*}
  B &=& \tau(P) = \tau(24) \cup \tau(4) \cup \tau(1) \cup \tau(0)\\
  &=& \{4,3\} \cup \{2\} \cup \{0\} \cup \{\varnothing\},\\
  E_* &=& E_{4,24} \cup E_{1,4} \cup E_{0,0}\\
   &=& \left\{e_1^{(4,24)}\right\} \cup \left\{e_1^{(2,4)}\right\} \cup \left\{e_1^{(0,0)}, e_2^{(0,0)} = *\right\},
\end{eqnarray*}
and the maps $\mathrm{pre}$ and $\mathrm{post}$ are defined as it shown in fig.\ref{last2}. We see that the Petri net $\mathscr{N}(P)$ can be described as ``an attaching'' the net $N_1$ to $N_2$ by the map $\varphi(b_{12}) = \varphi(b_{21}).$
\end{example}

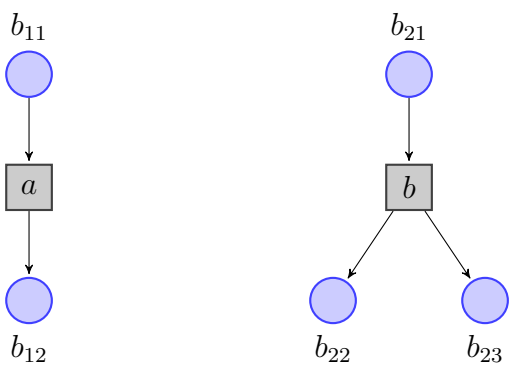
\begin{figure}[h!]
\begin{center}
\begin{tikzpicture}[node distance=1.3cm,>=stealth',bend angle=25,auto]
  \tikzstyle{place}=[circle,thick,draw=blue!75,fill=blue!20,minimum size=6mm]
  \tikzstyle{red place}=[place,draw=red!75,fill=red!20]
  \tikzstyle{transition}=[rectangle,thick,draw=black!75,
  			  fill=black!20,minimum size=6mm]
  \node[place,label=above:$b_{11}$] (1) at (0,0) {};
  \node[place,label=below:$b_{12}$] (2) at (0,-3) {};
  \node[place, label=above:$b_{21}$] (2') at (5,0) {};
  \node[place, label=below:$b_{22}$] (3) at (4,-3) {};
  \node[place, label=below:$b_{23}$] (4) at (6,-3) {};
  \node[transition] (a) at (0,-1.5) {$a$} edge [pre] (1) edge [post](2);
  \node[transition] (b) at (5,-1.5) {$b$} edge [post] (3) edge[post] (4) edge[pre] (2');
  \end{tikzpicture}
\end{center}
\caption{Two Petri nets $N_1$(left side) and $N_2$(right side) are shown.}\label{last1}\end{figure}

\begin{figure}[h!]
\begin{center}
\begin{tikzpicture}[node distance=1.3cm,>=stealth',bend angle=25,auto]
  \tikzstyle{place}=[circle,thick,draw=blue!75,fill=blue!20,minimum size=6mm]
  \tikzstyle{red place}=[place,draw=red!75,fill=red!20]
  \tikzstyle{transition}=[rectangle,thick,draw=black!75,
  			  fill=black!20,minimum size=6mm]
  \node[place,label=above:$1$] (1) at (0,0) {};
  \node[place,label=above:$2$] (2) at (3,0) {};
  \node[place, label=above:$3$] (3) at (6,1) {};
  \node[place, label=below:$4$] (4) at (6,-1) {};
  \node[transition] (a) at (1.5,0) {$e^{(1,4)}_1$} edge [pre] (1) edge [post](2);
  \node[transition] (b) at (4.5,0) {$e^{(4,24)}_1$} edge [pre] (2) edge [post](3) edge[post] (4);
  \node[transition] (c) at (7,0){$e^{(0,0)}_1$};
  \end{tikzpicture}
\end{center}
\caption{The Petri net $\mathscr{N}(P)$ is shown, where $P=xy^4+ x^4y^{24}+2$.}\label{last2}\end{figure}
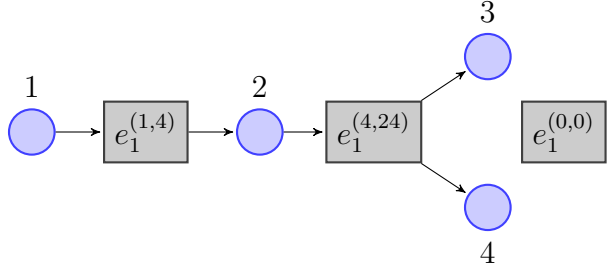

The preceding example shows that it is convenient to introduce the following definition.

\begin{definition}\label{sumofnets}
  Let $N_1 = (B_1, E_{1*},\mathrm{pre}_1, \mathrm{post}_1)$, $N_2 = (B_2, E_{2*},\mathrm{pre}_2, \mathrm{post}_2)$ be two Petri nets and let $\varphi_1:B_1 \to \mathbb{N}$, $\varphi_2:B_2 \to \mathbb{N}$ be injections. We define {\it an attaching $N_1$ to $N_2$ by $\varphi_1$, $\varphi_2$} as the following Petri net $N_1\bigsqcup\limits_{\varphi_1,\varphi_2}N_2 = (B,E_*,\mathrm{pre}, \mathrm{post})$, where $B = B_1 \sqcup B_2/\{ b_1 \sim b_2 \mbox{ if $\varphi_1(b_1) = \varphi_2(b_2)$}\}$, the events $E$ has the following form $E_{1*}\sqcup_* E_{2*} \cup \{\star\}$, here $E_{1*}\sqcup_* E_{2*}$ is the coproduct (wedge sum) in $\mathsf{Set}_*$, $\star$ is an added event and the maps $\mathrm{pre}$, $\mathrm{post}$ are defined as follows \[
  \mathrm{pre}(e) = \begin{cases}\mathrm{pre}_1(e), \mbox{ if } e \in E_{1}, \\ \mathrm{pre}_2(e), \mbox{ if } e \in E_{2}, \\ \end{cases} \quad \mathrm{post}(e) = \begin{cases}\mathrm{post}_1(e), \mbox{ if } e \in E_{1*}, \\ \mathrm{post}_2(e), \mbox{ if } e \in E_{2*}, \end{cases}
  \]
and we set $\mathrm{pre}(\star) = \mathrm{post}(\star) = \varnothing.$
\end{definition}

\begin{theorem}
Let $N_1 = (B_1, E_{1*},\mathrm{pre}_1, \mathrm{post}_1)$, $N_2 = (B_2, E_{2*},\mathrm{pre}_2, \mathrm{post}_2)$ be two Petri nets and let $\varphi_1:B_1 \to \mathbb{N}$, $\varphi_2:B_2 \to \mathbb{N}$ be injections. Consider $N_1 \bigsqcup\limits_{\varphi_1,\varphi_2} N_2 = (B,E_*,\mathrm{pre}, \mathrm{post})$ and let an injection $\varphi:B \to \mathbb{N}$ is defined as follows $\varphi(b) = \begin{cases} \varphi_1(b), \mbox{ if } b \in B_1, \\ \varphi_2(b), \mbox{ if } b \in B_2. \end{cases}$ Then $P(N,\varphi) = P(N_1,\varphi_1) + P(N_2,\varphi_2).$ Conversely, let $P_1,P_2 \in \mathbb{N}[x,y]$ be two polynomials such that $P_1(0,0), P_2(0,0)\ne 0$, then the Petri nets $\mathscr{N}(P_1 + P_2)$ and $\mathscr{N}(P_1) \bigsqcup\limits_{\iota_1,\iota_2} \mathscr{N}(P_2)$ are isomorphic in the category $\mathsf{NP}$.
\end{theorem}
\begin{proof}
From Construction \ref{construction1} and Definition \ref{sumofnets}  it follows that
\[
\sum\limits_{e \in E_*}x^{i(e)}y^{(j(e))} = \sum\limits_{e_1 \in E_{1*}} x^{i(e_1)}y^{j(e_1)} + \sum\limits_{e_2 \in E_2} x^{i(e_2)}y^{j(e_2)} + x^{i(\star)}y^{j(\star)},
\]
and since $\varphi(b) = \begin{cases} \varphi_1(b), \mbox{ if } b \in B_1, \\ \varphi_2(b), \mbox{ if } b \in B_2, \end{cases}$ we see that $P(N,\varphi) = P(N_1,\varphi_1) + P(N_2,\varphi_2).$

Conversely,  let $P_1,P_2 \in \mathbb{N}[x,y]$ be two polynomials, $P_1 = \sum\limits_{(i_1,j_1) \in I_1}a_{i_1,j_1}^{(1)}x^{i_1}y^{j_1}$, $P_2 = \sum\limits_{(i_2,j_2) \in I_2}a_{i_2,j_2}^{(2)}x^{i_2}y^{j_2}$, such that $P_1(0,0) \ne 0$ and $P_2(0,0)\ne 0$. Let $\mathscr{N}(P_1 + P_2) =  \bigl(\widetilde{B},\widetilde{E}_*, \widetilde{\mathrm{pre}}, \widetilde{\mathrm{post}}\bigr)$ and $\mathscr{N}(P_1) \bigsqcup\limits_{\iota_1,\iota_2} \mathscr{N}(P_2) = (B, E_*,\mathrm{pre},\mathrm{post})$.

By Construction \ref{con2} we get
\begin{eqnarray*}
  \widetilde{B} &=& \tau(P_1 + P_2) = \bigcup\limits_{(i,j) \in I_1 \cup I_2} \left\{ \tau(i), \tau(j) \right\}\\
   &=& \tau(P_1) \cup \tau(P_2) \cong \tau(P_1) \sqcup \tau(P_2)/ \{\tau(i_1) = \tau(i_2), \tau(j_1) = \tau(j_2)\},\\
   E_* &=& \bigcup_{\substack{(i_1,j_1) \in I_1 \\ (i_2, j_2) \in I_2}} \left\{ e_1^{(i_1,j_1)}, \ldots,e_{a_{i_1,j_1}^{(1)}}^{(i_1,j_1)}\right\} \cup \left\{ e_2^{(i_2,j_2)}, \ldots,e_{a_{i_2,j_2}^{(2)}}^{(i_2,j_2)}\right\}
\end{eqnarray*}
it follows that there exists a bijection $\beta: \widetilde{B} \cong B$ and a bijection $\eta: \widetilde{E_*} \cong E_*$, because Definition \ref{sumofnets} and Construction \ref{con2} yield $B = \tau(P_1) \sqcup \tau(P_2)/ \{\tau(i_1) = \tau(i_2), \tau(j_1) = \tau(j_2)\}$, $\widetilde{E_*} \cong E_{1*} \sqcup_* E_{2*} \cup \{e_{a_{0,0}^{(1)}}^{(0,0)}\}$ and $\widetilde{E_*} \cong E_{1*} \sqcup_* E_{2*} \cup \{e_{a_{0,0}^{(2)}}^{(0,0)}\}$. Finally, it is easy to see that under these bijections $\beta, \eta$ the corresponding diagrams of Remark \ref{Isom} are commutative. This completes the proof.

\end{proof}

\section{Zariski Topology on Petri Nets}

From the preceding section it follows that we can consider every polynomial $P \in \mathbb{N}[x,y]$, $P(0,0) \ne 0$ as a Petri net and vice versa. As well known, every commutative (semi)ring can be endowed with the Zariski topology.

Recall some basic definitions and concepts of Zariski topology. We essentially follow \cite{G}.

{\it A semiring} is a nonempty set $R$ on which operations of addition and multiplication have been defined such that the following conditions hold:
 \begin{itemize}
 \item[(1)] $(R,+)$ is a commutative monoid with identity element $0$;
 \item [(2)] $(R,\cdot)$ is a semigroup;
 \item[(3)] multiplications distributes over addition from either side;
 \item[(4)] $0r = r0 =0$ for all $r \in R$.
 \end{itemize}
Further, if $(R,\cdot)$ is a commutative semigroup and it contains identity element then the semiring $R$ is called unitary commutative (or $R$ is a commutative semiring with unit). We consider only commutative semirings with unit.

As in the case of (associative commutative) rings, we can introduce ideals of semirings. An {\it ideal} $I$ of a semiring $R$ is a nonempty subset of $R$ satisfying the following conditions: (1) if $a,b \in I$ then $a+b \in I$; (2) if $a \in I$ and $r\in R$ then $ra\in I$; (3) $I \ne R$.

Let $R$ be an a (commutative) semiring with unit, an ideal $\mathfrak{p} \subset R$ is called prime if and only if whenever $a\cdot b \in \mathfrak{p}$, for $a,b \in R$, we must have either $a \in \mathfrak{p}$ or $b \in \mathfrak{p}.$ The {\it spectrum} of $R$, denoted $\mathrm{Spec}(R)$, is the set of all prime ideals of $R$. The set $\mathrm{Spec}(R)$ can be equipped with the Zariski topology, for which the closed sets are the sets
\[
 V(I):=\bigl\{\mathfrak{p}\in \mathrm{Spec}(R) | I\subseteq \mathfrak{p}\bigr\}
\]
where $I$ is an ideal. Then, it is easy to see that:
\begin{itemize}
  \item[(1)] $V\left(\sum\limits_{\alpha \in A}I_\alpha\right) = \bigcap\limits_{\alpha \in A}V(I_\alpha)$, for every family $\{I_\alpha\}_{\alpha \in A}$ of ideals of $R$,
  \item[(2)] $V(I) \cup V(J) = V(IJ) = V(I \cap J)$ for every pair $I,J$ of ideals of $R$.
\end{itemize}

A basis for the Zariski topology can be constructed as follows. For $f\in R$, define $D_f:=\mathrm{Spec}(R)\setminus V((f))$. Then each $D_f$ is an open subset of $\mathrm{Spec}(R)$, and $\{D_f|f\in R\}$ is a basis for the Zariski topology.

We are now able to present the following

\begin{theorem}\label{general}
  Let $N \in \mathsf{PN}$ be a finite Petri net. Let $V(N)$ be a set $\{N' \in \mathsf{PN}\}$ of all decomposable Petri nets, such that $N'$ is a factor of $N$. The set $\mathscr{P}_{\varphi} = \bigl\{(N =   (B,E_*,\mathrm{pre},\mathrm{post}),\varphi)\bigr\}$ of all finite Petri nets $N  \in \mathsf{PN}$ with fixed injections $\varphi:B \to \mathbb{N}$, can be endowed with a topology (the Zariski topology) which is a collection of the closed subsets $V(N)$.
\end{theorem}
\begin{proof}
  It is obviously that this topology arises from the Zariski topology on $\mathrm{Spec}\,\mathbb{N}[x,y]$. Using Constructions \ref{construction1}, \ref{con2}, Theorem \ref{theorem} and Corollary \ref{corol1} we complete the proof.
\end{proof}

\begin{corollary}
  Every undecomposable Petri net $N\in\mathsf{PN}$ is closed point in the topological space $\mathscr{P}_\varphi$.
\end{corollary}
\begin{proof}
  As well known, closed points in the Zariski topology of $\mathbb{N}[x,y]$ correspondence to maximal ideals of $\mathbb{N}[x,y]$. Because $N$ is undecomposable then it is clear that the ideal $(P(N,\varphi))$ is maximal. This completes the proof.
\end{proof}

\end{document}